\documentclass[runningheads,a4paper]{llncs}
\usepackage[utf8]{inputenc}
\usepackage{microtype}
\usepackage{amsmath,amssymb,amsfonts}

\usepackage{amsthm}
\usepackage{macro}

\usepackage[final]{pdfpages} 

\DeclareMathOperator{\lcm}{lcm}
\newcommand{\unroll}[1]{\mathcal{U}(#1)}
\newcommand{\tree}[1]{\mathbf{\lowercase{#1}}}
\newcommand{\forest}[1]{\mathbf{\uppercase{#1}}}
\newcommand{\depth}[2][]{\mathrm{depth}_{#1}(#2)}
\newcommand{\cycle}{\ell}
\newcommand{\cut}[2]{\mathcal{C}(#1,#2)}

\newcommand{\dt}[1]{\mathcal{D}(#1)}

\newcommand{\rf}[1]{\mathcal{R}(#1)}

\newcommand{\ie}{\emph{i.e.}\@\xspace}

\newcommand{\oct}{\ge_{c}}
\newcommand{\ioct}{\le_{c}}

\newcommand{\otc}{\ge_{t}}

\newcommand{\setSize}[2]{(#1)(#2)}
\newcommand{\setDive}[2]{(#1)\{#2\}}

\usepackage[color=red!60,textsize=small]{todonotes}

\title{Injectivity of polynomials over finite\\discrete dynamical systems}
\author{
	Antonio E. Porreca
	\and
	Marius Rolland
}

\institute{
  Aix-Marseille Université, CNRS, LIS, Marseille, France\\
  \email{marius.rolland@lis-lab.fr}
}
\authorrunning{A.E. Porreca \and M. Rolland}

\date{\today}

\begin{document}
	\maketitle
	\begin{abstract}
        The analysis of observable phenomena (for instance, in biology or physics) allows the detection of dynamical behaviors and, conversely, starting from a desired behavior allows the design of objects exhibiting that behavior in engineering. The decomposition of dynamics into simpler subsystems allows us to simplify this analysis (or design). Here we focus on an algebraic approach to decomposition, based on alternative and synchronous execution as the sum and product operations; this gives rise to polynomial equations (with a constant side). In this article we focus on univariate, injective polynomials, giving a characterization in terms of the form of their coefficients and a polynomial-time algorithm for solving the associated equations.
	\end{abstract}
	
	\section{Introduction}
\label{sec:introduction}

Finite discrete-time dynamical systems (FDDS) are pairs $(X,f)$ where $X$ is a finite set of states and $f: X \to X$ is a transition function (when~$f$ is implied, we denote $(X,f)$ simply as $X$).
We find these systems in the analysis of concrete models such as Boolean networks~\cite{gershenson2004random_bn,automata_book}
and we can apply them to biology~\cite{thomas1990biological_feedback,thomas1973genetic_control_circuits,bernot2013computational_biology} to represent, for example, genetic regulatory networks or epidemic models. 
They also appear in chemistry~\cite{reaction_systems}, or information theory~\cite{gadouleau2011graph_entropy}.

In the following, we focus on deterministic FDDS, and often identify them with their transition digraphs, which have uniform out-degree one.
These graphs are also known as functional digraphs. 
Their general shape is a collection of cycles, where each node of a cycle is the root of a finite in-tree (a directed tree with the edges directed toward the root).
The nodes of the cycles are periodic states, while the others are transient states.

The set $(\mathbb{D}, +, \times)$ of FDDS up to isomorphism, with the alternative execution of two systems as \emph{addition} and the synchronous parallel 
execution of two FDDS as \emph{multiplication} is a commutative semiring~\cite{article_fondateur}. 
As all semirings, we can define a semiring of polynomials over~$\mathbb{D}$ and thus polynomial equations.

Although it has already been proven that general polynomial equations over~$\mathbb{D}$ (with variables on both sides of the equation) are undecidable, it is easily proved that, if one member of the equation is constant, then the problem is decidable (there is just a finite number of possible solutions)~\cite{article_fondateur}.
This variant of the problem is actually in $\NP$ because, since sums and products can be computed in polynomial time, we can just guess the values of the variables (their size is bounded by the constant side of the equation), evaluate the polynomial, and compare it with the constant side (for out-degree one graphs, isomorphism can be checked in linear time~\cite{testIsoLinear}).

However, more restricted equations are not yet classified in terms of complexity.
For example, we do not know if monomial equations of the form $AX = B$ are $\NP$-hard, in $\P$, or possible candidates for an intermediate difficulty. 
However, it has been shown that when~$A$ and~$B$ are connected and with a cycle of length one (i.e., they are trees with a loop, or fixed point, on the root), the equation $AX = B$ can be solved in polynomial time~\cite{article_arbre}.
It has also been proved that, if $X$ is connected, we can solve $AX^k = B$ in polynomial time with respect to the size of the inputs ($A$ and $B$) for all positive integer~$k$~\cite{kroot}.

These results share a common feature: all these equations have \emph{at most one solution}. 
It is natural, then, to investigate whether \emph{all} polynomial equations with at most one solution can be solved in polynomial time. 
However, we still lack a characterization for this type of equation; indeed, for the same polynomial on one side of the equation, if we change the constant side, the number of solutions can change.
For this reason, considering the case where the polynomial is injective seems to be a relevant starting point, requiring a characterization of this class of polynomials, at least for the univariate case.

In this paper, we focus on this class of polynomials. 
After giving the necessary preliminaries (Section~\ref{sec:preliminaries}), we first focus on the semiring of unrolls (equivalence classes of FDDS sharing the same transient behavior, introduced in~\cite{article_arbre}) and prove that all univariate polynomials over unrolls are injective and that every equation $\sum_{i=0}^{m} \unroll{A_i} \unroll{X}^i = \unroll{B}$ can be solved in polynomial time (Section~\ref{section:poly_unroll}).
Then, we prove that all univariate polynomial $P = \sum_{i=0}^{m} A_i X^i$ with at least one cancelable $A_i$ (for some $i > 0$) is injective, by using a proof technique implying the existence of a polynomial-time algorithm for solving the associated equations (Section~\ref{section:cond_suf_poly_FDDS_inj}).
Finally, we show that this is also a necessary condition for injectivity (Section~\ref{section:cond_nec_poly_FDDS_inj}).

%%% Local Variables:
%%% mode: LaTeX
%%% TeX-master: "main"
%%% End:

	\section{Preliminaries}
\label{sec:preliminaries}

In this paper we compose and decompose FDDS in terms of two algebraic operations.
Thee \emph{sum} of two FDDS 
consists in the mutually
exclusive alternative between their behaviors, while the \emph{product} is their synchronous parallel execution.
The set of FDDS taken up to isomorphism, with these two operations, forms a semiring~\cite{article_fondateur}.
This semiring is isomorphic to the semiring of functional digraphs with the operations of disjoint union and direct product.
We recall~\cite{livreGraphe} that the direct product of two graphs $A$ and $B$ is the graph $C$
having vertices $V(C) = V(A) \times V(B)$ and edges
\[
E(C) = \{((u,u'), (v,v')) \mid (u,v) \in E(A), (u',v') \in E(B)\}.
\]
Within the class of \emph{connected} FDDS, we further distinguish those having a (necessarily  unique) cycle of size $1$ (\ie, a fixed point in terms of dynamics) and refer to them as \emph{dendrons}. 

In this semiring, the structure of the product is particularly rich. 
For example, the semiring is not factorial, \ie, there exist irreducible FDDS $A,B,C,D$ such that $A \neq C$ and $A \neq D$ but $AB = CD$.
This richness might originate from the interactions between the periodic and the transient behaviors of the two FDDS being multiplied.

% For brevity, we use the term \emph{trees} in reference to in-trees constituting the transient behaviors of FDDS.

For the periodic behavior of FDDS, the product of a connected $A$, with a cycle of size $p$,
and a connected $B$, with a cycle of size $q$,
generates $\gcd(p,q)$ connected components, each with a cycle of size $\lcm(p,q)$~\cite{livreGraphe}.
For the analysis of transient behaviors, we use the notion of unroll introduced in~\cite{article_arbre}.

\begin{definition}[Unroll]
	Let $A = (X,f)$ be an FDDS.
	For each state $u\in X$ and $k \in \mathbb{N}$, we denote by $f^{-k}(u) = \{ v \in X \mid f^k(v) = u \}$ the set of $k$-th preimages of $u$. 
	For each $u$ in a cycle of $A$, we call the \emph{unroll tree of $A$ in $u$} the infinite tree $\tree{t}_u = (V,E)$ having vertices $V = \{(s,k) \mid s \in f^{-k}(u), k \in \mathbb{N}\}$ and edges
	$E = \big\{ \big((v,k),(f(v),k-1) \big) \big\} \subseteq V^2$.
	We call \emph{unroll of $A$}, denoted $\unroll{A}$, the set of its unroll trees (see Fig.~\ref{fig:unroll}).
\end{definition}

\begin{figure}[t]
\centering
\includegraphics[page=1]{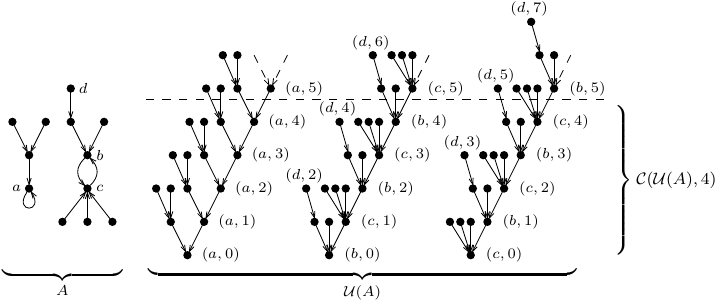}
\caption{An FDDS~$A$ with two connected components (on the left), a finite portion of its unroll~$\unroll{A}$ (on the right) and the cut~$\cut{\unroll{A}}{4}$ (below the dashed line). A few vertex names are shown in order to highlight their contribution to the unroll.}
\label{fig:unroll}
\end{figure}

In the rest of this paper, \emph{unrolls will always be taken up to isomorphism}, \ie, as multisets or sums (disjoint unions) of unlabeled trees; the equality sign will thus represent the isomorphism relation.

An unroll tree contains exactly one infinite branch, onto which are periodically rooted the trees representing the transient behaviors of the corresponding connected component. 
% In order to employ unrolls to study the transient behaviors with respect to the product,
We also exploit a notion of ``levelwise'' product on trees. 
For readability, we will denote trees and forests using bold letters (in lower and upper case
respectively) to distinguish them from FDDS.

\begin{definition}[Product of trees]\label{prodintrees} 
	Let $\tree{t}_1=(V_1,E_1)$ and $\tree{t}_2=(V_2,E_2)$ be two trees with roots $r_1$ and $r_2$, respectively.
        Their \emph{product} is the tree $\tree{t}_1 \times \tree{t}_2=(V,E)$ with vertices
        $V=\set{(u,v)\in V_1\times V_2 \mid \depth{u}=\depth{v}}$, where $\depth{u}$ is the length of the shortest path between $u$ and the root of its tree, and edges $E=\set{((u,u'),(v,v')) \mid (u,v)\in E_1, (u',v')\in E_2}\subseteq V^2$ (see Fig.~\ref{fig:tree-product}).
\end{definition}

\begin{figure}[t]
\centering
\includegraphics[page=2]{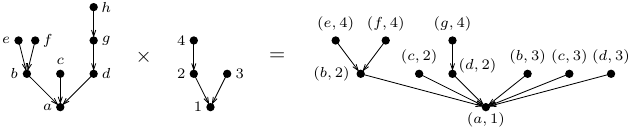}
\caption{The levelwise product of two finite trees. Remark how the depth of the result is given by the minimum depth of the two factors.}
\label{fig:tree-product}
\end{figure}

We deduce from~\cite{article_arbre} that the set of unrolls up to isomorphism, with disjoint union for addition and product of trees for multiplication, is a semiring with the infinite path as the multiplicative identity.
In addition, the unroll operation is a homomorphism between the semiring~$\mathbb{D}$ of FDDS and the semiring of unrolls.
% Here and in the following, the equality sign will denote graph isomorphism.

Unrolls have one major algorithmic issue: they are infinite. 
To overcome this obstacle, we will consider only a finite portion of the unrolls. 
For this purpose, % we first define the \emph{depth} of a finite tree as the maximum length of a shortest path between a leaf and its root.
we extend the notion of \emph{depth} to forests of finite trees by taking the maximum depth of its trees. 
We also define the \emph{depth of an FDDS} as the maximum depth among the trees rooted in one of its periodic states.

We remark that the set of forests of bounded depth~$d$ up to isomorphism is a also semiring (with the restrictions of the same tree sum and product operations) for all depth~$d$, where the multiplicative identity is the path of depth~$d$.

We can now define the \emph{cut} of $\tree{t}$ at depth $k \ge 0$, denoted $\cut{\tree{t}}{k}$,
as the induced sub-tree of $\tree{t}$ restricted to the vertices of depth lesser or equal to $k$ (see Fig.~\ref{fig:unroll}). 
This operation generalizes (homomorphically) to forests $\forest{f} = \tree{t}_{1} + \ldots + \tree{t}_{n}$ as  $\cut{\forest{F}}{k} = \cut{\tree{t}_{1}}{k} + \ldots + \cut{\tree{t}_{n}}{k}$. 
The cut of an unroll at a certain depth has already been used to exhibit properties of unrolls~\cite{article_arbre,kroot}.
This approach notably gave a characterization of \emph{cancelable} FDDS, \ie FDDS $A$ such that $AB = AC$ implies $B = C$ for all FDDS $B,C$.
Indeed, a FDDS $A$ is cancelable if and only if at least one of its connected components is a dendron~\cite[Theorem 30]{article_arbre}. 

One of the tools for the analysis of trees and forests is the total order $\le$ on finite and infinite trees introduced in~\cite{article_arbre}.
Indeed, this order is compatible with the product for infinite trees\footnote{This compatibility with the product of~$\le$ is sufficient for most applications in this paper; we refer the reader to the original paper for the actual definition.}, \ie, if $\tree{t}_1, \tree{t}_2$ are two infinite trees then $\tree{t}_1 \le \tree{t}_2$ if and only if $\tree{t}_1 \tree{t} \le \tree{t}_2 \tree{t}$ for all tree $\tree{t}$~\cite[Lemma 24]{article_arbre}. 
A similar property has been proved for the case of finite trees:

\begin{lemma}[\cite{article_arbre}]\label{lemma:ordre_comp_produit}
	Let $\tree{t}_1, \tree{t}_2, \tree{t}$ then $\cut{\tree{t}_1}{\depth{\tree{t}}} \le \cut{\tree{t}_2}{\depth{\tree{t}}}$ if and only if $\tree{t}_1 \tree{t} \le \tree{t}_2 \tree{t}$.
\end{lemma}

In the rest of the paper all polynomials will be implicitly univariate and, by convention, the symbol~$\forest{a}^0$ (i.e., the $0$-th power of a forest of finite trees~$\forest{a}$) will denote the path of length~$\depth{\forest{a}}$, which behaves as an identity for the product with forests of equal or less depth.

%%% Local Variables:
%%% mode: LaTeX
%%% TeX-master: "main"
%%% End:

	\section{Characterization of injective polynomials over unrolls}\label{section:poly_unroll}
	
	The search for a solution to polynomial equations over FDDS can be seen as the search for a compatible solution between solutions of polynomial equations of transient states and solutions of polynomial equations of periodic states.
	This is why we can start by looking for solutions for transient states.
	In order to do this, we will use the notion of unroll and determine the number of solutions to polynomial equations over unrolls.
	
	More precisely, our final goal is to prove the following theorem:
	\begin{theorem}\label{th:injPolyUnrolls}
		All univariate polynomial over unrolls are injective. 
	\end{theorem}
	
	For the proof of this theorem, we begin by remarking that for all polynomial~$P$ over unrolls, and for all FDDS $X,Y$ the equality $P(\unroll{X}) = P(\unroll{Y})$ implies $\cut{P(\unroll{X})}{n} = \cut{P(\unroll{Y})}{n}$ for all $n$. In order to study polynomial equations over forests of finite trees, we start by proving one technical lemma on the behavior of the order on powers of trees with the same depth.

		\begin{lemma}\label{lemma:behaviorOrderWithProduct}
			Let $\tree{x}, \tree{y}$ be two finite trees with the same depth and $k$ be a  positive integer. Then $\tree{y}^k \ge \tree{x}^k$ if and only if $\tree{y} \ge \tree{x}$.
		\end{lemma}
	
		\begin{proof}
			If $\tree{y} \ge \tree{x}$ and $\depth{\tree{x}} = \depth{\tree{y}}$, by \cite[Corollary 20]{article_arbre}, we have $\tree{y}^2 \ge \tree{x}^2$ and $\depth{\tree{y}^2} = \depth{\tree{x}^2}$. 
			Then, by induction, we obtain $\tree{y}^k \ge \tree{x}^k$.
			
			For the other direction, assume that $\tree{y}^k \ge \tree{x}^k$ and, by contradiction, that $\tree{x} > \tree{y}$. 
			Since $\depth{\tree{x}} = \depth{\tree{y}}$, by \cite[Corollary 20]{article_arbre}, we have $\tree{y}^k \tree{x} \ge \tree{y}^k \tree{y}$. 
			Thus, by \cite[Lemma 21]{article_arbre}, we have $\tree{y}^{k-1} \ge \tree{x}^{k-1}$. 
			By induction, we obtain $\tree{y} \ge \tree{x}$, a contradiction.
	\end{proof}
		
		Let $\forest{A}$ be a forest and $i\ge0$ be an integer. 
		We denote by % $\gamma(\forest{A},i)$ the sub-multiset of trees of $\forest{A}$ having depth \emph{exactly} $i$; we also
	        % denote by
                $\Gamma(\forest{A},i)$ the sub-multiset of trees of $\forest{A}$ where each tree has depth \emph{at least} $i$.
		We focus our attention on the trees in $P(\forest{X})$ with a
		certain depth and form.
		
		\begin{lemma}\label{lemme:ordreProdFin}
	
			Let $P = \polyF{1}{m}{a}{x}$ be a polynomial without constant term over forests and $\forest{X}= \tree{x}_1 + \ldots + \tree{x}_n$ be a forest sorted in nondecreasing order.
			Let $\forest{A}_i = \sum_{j = 1}^{m_i} \tree{a}_{i,j}$ be in nondecreasing order, hence $\tree{a}_{i,1} = \min(\forest{a}_i)$, for all~$i$.
			Then, there exists $\alpha \in \{1, \ldots, m\}$ such that $\min_{i}(\forest{A}_i \forest{X}^{i-1} \tree{x}_k)$ is isomorphic to $\tree{a}_{\alpha,1} \tree{x}_1^{\alpha-1}\tree{x}_k$.
		\end{lemma}
		
		\begin{proof}
			Let $\tree{t}_{min}$ be the smallest tree in $P(\forest{X})$. 
			Hence, there exists $\alpha \in \{1, \ldots, m\}$ such that $\tree{t}_{min} \in \forest{A}_\alpha \forest{X}^\alpha$. 
			In addition, by~\cite[Lemma 24]{article_arbre}, we have $\tree{x}_1^2 \le \tree{x}_1 \tree{x}_j$ for all $j$. 
			By induction, we deduce that $\tree{t}_{min} \in \forest{A}_\alpha \tree{x}_1^\alpha$. 
			Analogously, we have $\tree{a}_{\alpha,1} \tree{x}_1^\alpha \le \tree{a}_{\alpha,j} \tree{x}_1^\alpha$ for all $\tree{a}_{\alpha,j} \in \forest{A}_\alpha$, and we conclude that $\tree{t}_{min} = \tree{a}_{\alpha,1} \tree{x}_1^\alpha$.
			
			We recall that $\tree{t}_1 \tree{t}_3 < \tree{t}_2 \tree{t}_3$ if and only if $\cut{\tree{t}_1}{\depth{\tree{t}_3}} < \cut{\tree{t}_2}{\depth{\tree{t}_3}}$ (Lemma~\ref{lemma:ordre_comp_produit}). 
			Therefore, for all $j$ and $\tree{a} \in \forest{A}_j$, we have:
			\[\tree{a}_{\alpha,1} \tree{x}_1^\alpha \le \tree{a} \tree{x}_1^j \Leftrightarrow \cut{\tree{a}_{\alpha,1} \tree{x}_1^{\alpha-1}}{\depth{\tree{x}_1}} \le \cut{\tree{a} \tree{x}_1^{j-1}}{\depth{\tree{x}_1}}\]
			Three cases are possible.
			First, if  $\alpha = j$ then $\tree{a}_{\alpha,1} \le  \tree{a}$ by minimality of $\tree{a}_{\alpha,1}$.
			This implies that $\tree{a}_{\alpha,1} \tree{x}_1^{\alpha-1} \tree{x}_k \le \tree{a} \tree{x}_1^{\alpha-1} \tree{x}_k$.
	
			Second, if~$\alpha \ne j$ and $\alpha \neq 1$, then $\cut{\tree{a}_{\alpha,1} \tree{x}_1^{\alpha-1}}{\depth{\tree{x}_1}} = \tree{a}_{\alpha,1} \tree{x}_1^{\alpha-1}$. 
			Since $\cut{\tree{a} \tree{x}_1^{j-1}}{\depth{\tree{x}_1}} \le\tree{a} \tree{x}_1^{j-1}$, we have $\tree{a}_{\alpha,1} \tree{x}_1^{\alpha-1} \tree{x}_k \le \tree{a} \tree{x}_1^{j-1} \tree{x}_k$. 
			
			Thirdly, if $\alpha \ne j$ but $\alpha = 1$, then either
                        \[
                        \cut{\tree{a}_{\alpha,1}\tree{x}_1^{\alpha-1}}{\depth{\tree{x}_1}} = \cut{\tree{a} \tree{x}_1^{j-1}}{\depth{\tree{x}_1}} = \tree{a} \tree{x}_1^{j-1}
                        \]
                        which implies $\tree{a}_{\alpha,1}\tree{x}_1^{\alpha} = \tree{a}_{j,1}\tree{x}_1^{j}$, and we can just choose $\alpha=j$, which is already covered by the second case. 
			Otherwise $\cut{\tree{a}_{\alpha,1}\tree{x}_1^{\alpha-1}}{\depth{\tree{x}_1}} < \tree{a} \tree{x}_1^{j-1}$. 
			In this case, we need to recall the definition of the order from~\cite{article_arbre}~and the code upon which it is defined. 
			Indeed, by~\cite[Lemma 17]{article_arbre}, we know that the leftmost difference between the codes of $\cut{\tree{a}_{\alpha,1}\tree{x}_1^{\alpha-1}}{\depth{\tree{x}_1}}$ and $\cut{\tree{a} \tree{x}_1^{j-1}}{\depth{\tree{x}_1}}$ corresponds to a node of depth at most $\depth{\tree{x}_1} - 1$. 
			Therefore, this difference is propagated into $\cut{\tree{a}_{\alpha,1}\tree{x}_1^{\alpha-1}}{\depth{\tree{x}_1}} \tree{x}_k$ and $\tree{a} \tree{x}_1^{j-1} \tree{x}_k$ if $\depth{\tree{x}_k} \ge \depth{\tree{x}_1}-1$, and $\tree{a}_{\alpha,1} \tree{x}_1^{\alpha-1} \tree{x}_k < \tree{a} \tree{x}_1^{j-1} \tree{x}_k$; if $\depth{\tree{x}_k} < \depth{\tree{x}_1}-1$ then $\tree{a}_{\alpha,1} \tree{x}_1^{\alpha-1} \tree{x}_k = \tree{a} \tree{x}_1^{j-1} \tree{x}_k$.
			
			Hence, in all cases we have $\tree{a}_{\alpha,1} \tree{x}_1^{\alpha-1} \tree{x}_k \le \tree{a} \tree{x}_1^{j-1} \tree{x}_k$.
			Since for all $p$ we have $\tree{x}_1^{p-1} \le \tree{t}$ with $\tree{t} \in \forest{X}^{p-1}$, the lemma follows.
		\end{proof}
		
		Thanks to Lemma~\ref{lemme:ordreProdFin} we can prove the first important result of this section. Remark that, due to closure under sums and products, the set of forests of finite trees of depth at most~$d_{max}$ is a semiring for every natural number~$d_{max}$.
		
		\begin{proposition}\label{prop:injDesFinis}
	        Let~$P = \sum_{i=0}^{m} \forest{A}_i \forest{X}^i$ be a polynomial with~$d_{max} = \max_{i > 0}(\depth{\forest{A}_i)}$ and~$\depth{\forest{A}_0} \le d_{max}$. Then~$P$ is injective over the semiring of forests of finite trees of depth at most~$d_{max}$. 
	        
	        % Let~$\forest{A}_0$ and~$\forest{A}_1, \ldots, \forest{A}_m$ be forests of finite trees with~$d_{max} = \max_{i > 0}(\depth{\forest{A}_i)}$ and~$\depth{\forest{A}_0} \le d_{max}$. Consider the polynomial~$P = \sum_{i=0}^{m} \forest{A}_i \forest{X}^i$ over the semiring of forests of finite trees of depth at most~$d_{max}$. Then~$P$ is injective.
			% Let $d_{max}$ be a natural number and $P = \sum_{i=1}^{m} \forest{A}_i \forest{X}^i$ a polynomial over forests without constant terms with depth at most $d_{max}$\mr{formulation?} such that $d_{max} = \max(\depth{\forest{A}_i})$.
			% Then $P$ is injective.
	%		Then, for all $0\le d \le d_{max}$ then $\sum_{i=1}^{m} \Gamma(\forest{A}_i,d) \Gamma(\forest{X}^i, d)$ is injective.
		\end{proposition}

		\begin{proof}
	        Let us first consider polynomials~$P$ without a constant term.
			Let $\forest{X} = \tree{x}_1 + \ldots + \tree{x}_{m_x}$ and $\forest{Y} = \tree{y}_1 + \ldots + \tree{y}_{m_y}$ be two forest (sorted in nondecreasing order), and assume $P(\forest{X}) = P(\forest{Y})$. 
			Thus, by Lemma~\ref{lemme:ordreProdFin}, there exist $i, j \in \{1, \ldots, m\}$ such that, for each $k$, the smallest tree $\tree{t}_s$ with factor $\tree{x}_k$ (resp., $\tree{y}_k$) in $\polyF{1}{m}{a}{x}$ is isomorphic to $\tree{a}_{i,1} \tree{x}_1^{i-1}\tree{x}_k$ (resp., $\tree{a}_{j,1} \tree{y}_1^{j-1}\tree{y}_k$), with $\tree{a}_{i,1} = \min(\forest{A}_i)$ and $\tree{a}_{j,1} = \min(\forest{A}_j)$. 
			From the hypothesis, we deduce that $\tree{a}_{i,1} \tree{x}_1 = \min(P(\forest{x})) =  \min(P(\forest{y})) = \tree{a}_{j,1} \tree{y}_1$. 
			
			First, we show that $\tree{x}_1 = \tree{y}_1$.
			By contradiction, and without loss of generality, assume $\tree{x}_1 < \tree{y}_1$.
			Then, by Lemma~\ref{lemma:behaviorOrderWithProduct}, it follows that $\tree{x}_1^{j} < \tree{y}_1^{j}$.
			By \cite[Lemma 24]{article_arbre}, we deduce that $\tree{a}_{j,1} \tree{x}_1^{j} < \tree{a}_{j,1} \tree{y}_1^j = \tree{a}_{i,1} \tree{x}_1^i$.
			Nevertheless, $\tree{a}_{j,1} \tree{x}_1^{j} \in P(\forest{X})$,
			contradicting the minimality of $\tree{a}_{i,1} \tree{x}_1^i$. 
			
			Finally, since $\tree{x}_1 = \tree{y}_1$, we deduce that $P(\forest{X}) - P(\tree{x}_1) = P(\forest{Y}) - P(\tree{y}_1)$.  
			In addition, the smallest tree in $P(\forest{X}) - P(\tree{x}_1)$ is $\tree{a}_{i,1} \tree{x}_1^{i-1} \tree{x}_2$ and the smallest tree in $P(\forest{y}) - P(\tree{y}_1)$ is $\tree{a}_{j,1} \tree{y}_1^{j-1} \tree{y}_2$.  
			However, since $\tree{a}_{i,1} \tree{x}_1^i = \tree{a}_{j,1} \tree{y}_1^j$ and $\tree{x}_1 = \tree{y}_1$, we have $\tree{a}_{i,1} \tree{x}_1^{i-1} = \tree{a}_{j,1} \tree{y}_1^{j-1}$. 
			Thus, $\tree{x}_2 = \tree{y}_2$. 
			By induction, we conclude that $\forest{X} = \forest{Y}$.
	
	        This proof extends to polynomials with a constant term~$\forest{A}_0$. Indeed, let~$Q = P + \forest{A}_0$ and suppose~$Q(\forest{X}) = Q(\forest{Y})$. Then $Q(\forest{X}) - \forest{A}_0 = Q(\forest{Y}) - \forest{A}_0$, that is~$P(\forest{X}) = P(\forest{Y})$, which implies~$\forest{X} = \forest{Y}$. Thus, $Q$ is also injective.
		\end{proof}
	
	        Our main theorem follows directly from Proposition~\ref{prop:injDesFinis}.
	
	   \begin{proof}[Proof of Theorem~\ref{th:injPolyUnrolls}]
	        Suppose~$P(\unroll{X}) = P(\unroll{Y})$. Then~$\cut{P(\unroll{X})}{n} = \cut{P(\unroll{Y})}{n}$ for all~$n$. If~$\unroll{X}$ and~$\unroll{Y}$ were different, there would exist an~$n$ such that~$\cut{\unroll{X}}{n} \ne \cut{\unroll{Y}}{n}$. However, by Proposition~\ref{prop:injDesFinis} and since~$\cut{\cdot}{n}$ is a homomorphism, we have~$\cut{\unroll{X}}{n} = \cut{\unroll{Y}}{n}$ for all~$n$.
	    \end{proof}
	        
		% From this theorem, we can draw two remarks. Indeed, by Proposition~\ref{prop:polyUnroll2PolyForestFini}, we have that the polynomial $\sum_{i=1}^{m} \unroll{A_i} \unroll{X}^i$ is injective if and only if $\sum_{i=1}^{m} \cut{\unroll{A_i}}{n} \cut{\forest{y}}{n}^i$ is injective for a certain positive integer $n$.
		% And, for this kind of polynomial, each forest has depth at most $n$.  
		% Then, by Proposition~\ref{prop:injDesFinis}, the resulting function is injective.
	        % and the proof of Theorem \ref{th:injPolyUnrolls} follows.
		
		The proof of Proposition~\ref{prop:injDesFinis} also suggests an algorithm (inspired by \cite[Algorithm~1]{kroot}) for solving polynomial equations over forests of finite trees. 
		Indeed, starting with the maximum depth $d_{max}$ (Fig.~\ref{fig:algo} \textcircled{1}), we can solve the equation $\Gamma(P(\forest{x}), d_{max}) = \Gamma(\forest{b}, d_{max})$ \textcircled{2}.
		In fact, we construct the solution $\Gamma(\forest{X},d_{max})$ by first solving the equation $\min(\Gamma(\forest{A}_i,d_{max})) \tree{x}_1^i = \min(\Gamma(\forest{b}, d_{max}))$ \textcircled{3}\textcircled{4} for an~$i \in \{1, \ldots, m\}$. 
		The tree~$\tree{x}_1$ is minimal among the trees of maximum depth appearing in the solution~$\forest{X}$. 
		Thus, by Lemma~\ref{lemme:ordreProdFin}, either we can construct the solution $Sol$ inductively \textcircled{4} or $P(Sol) \nsubseteq \forest{B}$ \textcircled{3}, and we can deduce that this~$i$ is not the one appearing in the statement of Lemma~\ref{lemme:ordreProdFin}, and we try the next value of~$i$, if any.
		Indeed, the $j$-th tree~$\tree{x}_i$ of the solution is equal to
	        \[
	        \cfrac{\min(\Gamma(\forest{b}, d_{max}) - \Gamma(P(\tree{x}_1 + \ldots + \tree{x}_{j-1}), d_{max}))}{\min(\Gamma(\forest{A}_i,d_{max})) \tree{x}_1^{i-1}}.
	        \]
     	Then, we can inductively construct $\tree{x} = \Gamma(\forest{X}, d)$ from $\Gamma(\forest{x}, d')$ where $d < d_{max}$ is the depth of a tree in $\forest{B}$ and $d'$ the smallest depth of a tree in $\forest{B}$ greater than~$d$~\textcircled{5}. In order to do so, we need to find~$\min(\Gamma(\forest{X}, d))$. Let $\tree{a} = \min(\Gamma(\forest{A}_j,d))$ for a~$j \in \{1, \ldots, m\}$ and~$\tree{b} = \min(\Gamma(\forest{B},d))$.
     	\begin{itemize}
	        \item If $\tree{b}$ has depth at least~$d'$ \textcircled{6}, or $\tree{a}$ has depth~$d$ and~$\tree{a} \times \min(\Gamma(\forest{X}, d'))^j \le \tree{b}$, then $\tree{x}$ has depth at least~$d'$ and thus is already in $\Gamma(\forest{x}, d')$.
	        \item Otherwise~$\tree{b}$ necessarily has depth~$d$. In that case, either~$\tree{a}$ has depth at least~$d'$, or~$\tree{a} \times \min(\Gamma(\forest{X}, d'))^j > \tree{b}$, and in both cases we have~$\tree{x} = \sqrt[j]{\tree{b} / \tree{a}}$.
     	\end{itemize}
		Then, similarly to the step for depth~$d_{max}$, we can construct a portion of the solution~$\forest{X}$ \textcircled{8}, or find a mistake \textcircled{7} and try another value of~$j$, if any. If during the iteration we find mistakes for all coefficients, this implies that the equation does not have a solution.
		
		Moreover, the number of iterations is bounded by the number of coefficients~$m$ times the depth~$d_{max}$, which are both polynomial with respect to the size of the inputs, and the operations carried out in each iteration can be performed in polynomial time \cite{article_arbre,kroot}.
	
	        We cannot directly apply this algorithm to polynomial equations over unrolls, since they contain infinite trees, but we can show that, in fact, we only need to consider a finite cut at depth polynomial with respect to the sizes of the FDDS in order to decide if a solution exists:
	
	% \begin{proposition}\label{prop:polyUnroll2PolyForestFini}
	% 	Let $P = \sum_{i=0}^{m} A_i X^i$ be polynomial over FDDS, $B$ an FDDS and $\alpha$ the number of unroll trees of $\unroll{B}$.
	% 	Let $n \ge 2 \times \alpha^2 + \depth{\unroll{B}}$ be an integer.
	% 	If there exists a forest $\forest{y}$ with $\depth{\forest{y}} \le n$ such that $\sum_{i=0}^{m}\cut{\unroll{A_i}}{n} \forest{y}^i = \cut{\unroll{B}}{n}$, then there exists a FDDS $X$ such that $\unroll{P(X)} = \unroll{B}$ and $\cut{\unroll{X}}{n} = \forest{y}$.
	% \end{proposition}

        \begin{figure}[p]
        \centering
        \includegraphics[page=3,width=\textwidth]{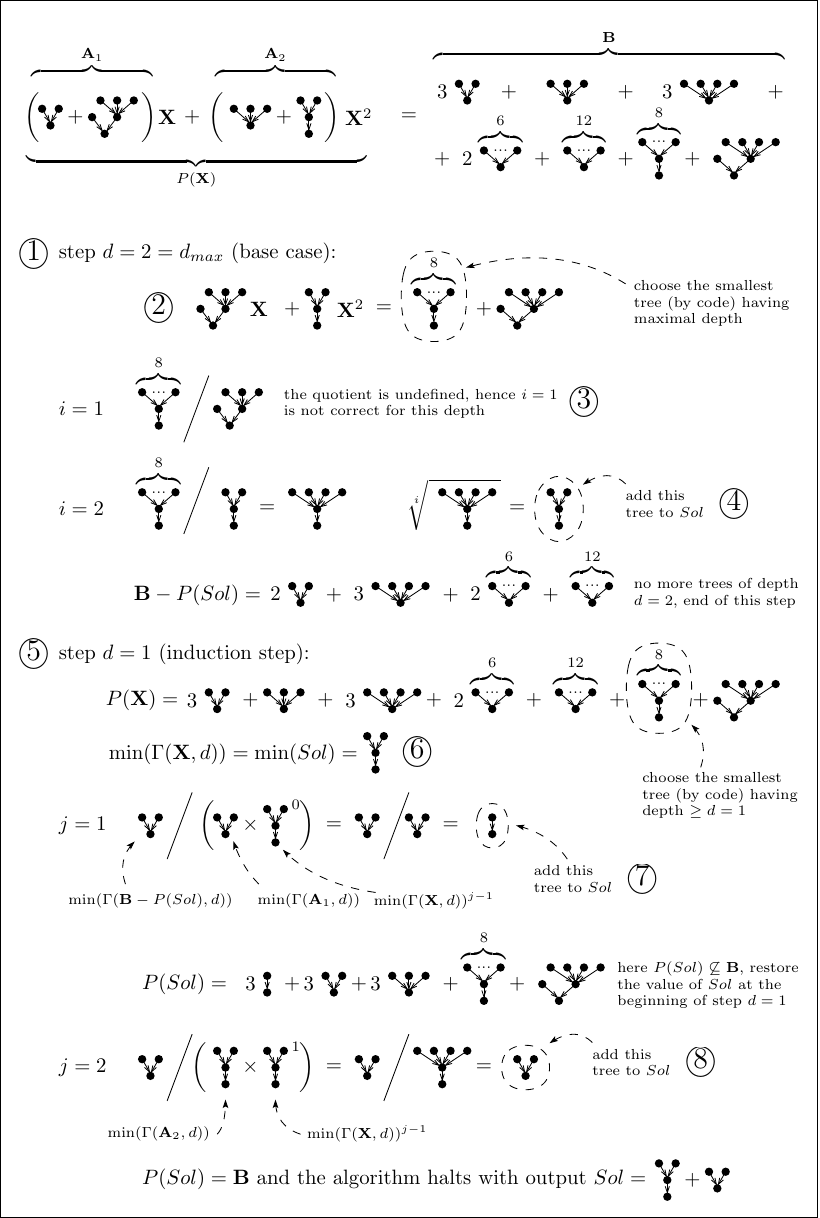}
        % \vspace{-1em}
        \caption{A run of the algorithm for polynomial equations over forests of finite trees.}
        \label{fig:algo}
        \end{figure}

	\begin{proposition}\label{prop:polyUnroll2PolyForestFini}
		Let $P = \sum_{i=0}^{m} A_i X^i$ be polynomial over FDDS, $B$ an FDDS and $\alpha$ the number of unroll trees of $\unroll{B}$.
		Let $n \ge 2 \times \alpha^2 + \depth{\unroll{B}}$ be an integer.
		Then, there exists a forest $\forest{y}$ with $\depth{\forest{y}} \le n$ such that $\sum_{i=0}^{m}\cut{\unroll{A_i}}{n}  \forest{y}^i = \cut{\unroll{B}}{n}$ if and only if there exists a FDDS $X$ such that $\unroll{P(X)} = \unroll{B}$. Furthermore, if such an $X$ exists, then necessarily $\cut{\unroll{X}}{n} = \forest{y}$.
	\end{proposition} 
	
	\begin{proof}
		Assume that $\sum_{i=0}^{m}\cut{\unroll{A_i}}{n} \forest{y}^i = \cut{\unroll{B}}{n}$.
		First, remark that we can apply a similar reasoning to \cite[Lemma 11]{kroot} in order prove that if there exists $\forest{x}$ such that $\sum_{i=0}^{m} \unroll{A_i} \forest{X} = \unroll{B}$ then there exists an FDDS $Y$ such that $\unroll{Y} = \forest{X}$. 
		Since $\alpha$ bounds the greatest cycle length in $A_0,\ldots,A_m,B$ and since, from the hypothesis, $\depth{\unroll{B}}$ bounds $\depth{\unroll{A_i}}$ for each $i$, we have all the tools needed in order to prove (similarly to \cite[Theorem 5]{kroot}) that there exists a FDDS $X$ such that $\unroll{P(X)} = \unroll{B}$. The inverse implication is trivial.
        \end{proof}
                
        This proposition has two consequences. 
        First of all, the solution (if it exists) to an equation $\sum_{i=0}^{m} \unroll{A_i} \forest{X}^i = \unroll{B}$ is always the unroll of an FDDS. 
        Furthermore, by employing the same method for rebuilding $\unroll{X}$ from the solution~$\forest{Y}$ with~$\depth{\forest{Y}} \le n$ of $\sum_{i=0}^{m}\cut{\unroll{A_i}}{n} \forest{y}^i = \cut{\unroll{B}}{n}$ exploited in the proof of \cite[Theorem 6]{kroot}, we deduce an algorithm for solving equations over unrolls. In order to encode unrolls (which are infinite graphs), we represent them finitely by FDDS having those unrolls.

		\begin{theorem}\label{theorem:equation_forest_P}
			Let $P$ be a polynomial over unrolls (encoded as a polynomial over FDDS) and~$B$ an FDDS.
			Then, we can solve the equation~$P(\unroll{X}) = \unroll{B}$ in polynomial time.
		\end{theorem} 	
	
%%% Local Variables:
%%% mode: LaTeX
%%% TeX-master: "main"
%%% End:

	\section{Sufficient condition for the injectivity of polynomials}\label{section:cond_suf_poly_FDDS_inj}

	Even if all polynomials over unrolls are injective, we can easily remark that is not the case for the FDDS.
	For example, let $A$ be the cycle of length $2$; then, the polynomial $AX$ is not injective since $A^2=2A$.
	Hence, since the existence of injective polynomials is trivial (just consider the polynomial $X$), we want to obtain a characterization of this class of polynomials. 
	% In order to do this, we begin by considering polynomials $P = \sum_{i=1}^{m} A_i X^i$ without constant terms.
	
	Proposition 40 of \cite{article_arbre} gives a sufficient condition for the injectivity of polynomials without constant term, namely if the coefficient~$A_1$ is cancelable, \ie, at least one of its connected component is a dendron.
	However, since~$X^k$ is injective for all~$k>0$ (by uniqueness of $k$-th roots~\cite{article_arbre}), this condition is not necessary. 
	Hence, we start by giving a sufficient condition which covers these cases.
	
	\begin{proposition}\label{prop:condSufInj}
		Let $P = \sum_{i=0}^{m} A_i X^{i}$ be a polynomial.
		If at least one coefficient of $P$ is cancelable then $P$ is injective.
	\end{proposition}

	For the proof of this proposition, we start by showing the case of polynomials $P = \sum_{i=1}^{m} A_i X^i$ without a constant term. 
	For that, we need to analyze the structure of the possible evaluations of $P$, and more precisely their multisets of cycles.
	We denote by $\cycle(A)$ the length of the cycle of a connected FDDS $A$.
	% For this analyze, we need two tools. 
	Let us introduce two total orders over connected FDDS:
	
	\begin{definition}
		% Let $a,b$ be two positive integer, $\tree{t},\tree{t'}$ be two unroll trees and
	        Let $A,B$ be two connected FDDS, let~$a = \cycle(A)$, $b = \cycle(B)$, $\tree{t}= \min(\unroll{A})$, and $\tree{t'} = \min(\unroll{B})$.
		We define:	
		\begin{itemize}
			\item $A \oct B$ if and only if $a > b$ or $a = b$ and $\tree{t} \ge \tree{t'}$.
			\item $A \otc B$ if and only if $\tree{t} > \tree{t'}$ or $\tree{t} = \tree{t'}$ and $a \ge b$.
		\end{itemize}
	\end{definition}
	
	% The second tool for our analysis is two functions of connected component.
	For an FDDS $X$ and a positive integer~$p$, let us consider the function~$\setSize{X}{p}$, which computes the multiset of connected components of $X$ with cycles of length $p$, and the function~$\setDive{X}{p}$, computing the multiset of connected components of $X$ with cycles of length \emph{dividing} $p$.
	
	\begin{lemma}\label{lemme:closureProduct}
		Let $p > 0$ be an integer; then
		$\setDive{\cdot}{p}$ is an endomorphism over FDDS. 
	\end{lemma}
	
	\begin{proof}
		First, if we denote by $\mathbf{1}$ the fixed point (the multiplicative identity in~$\mathbb{D}$), we obviously have $\setDive{\mathbf{1}}{p} = \mathbf{1}$. 
		Now, let $A,B$ be two FDDS. 
		It is clear that $\setDive{A + B}{p} = \setDive{A}{p} + \setDive{B}{p}$ since the sum is the disjoint union.
		All that remains to prove is that $\setDive{AB}{p} = \setDive{A}{p} \setDive{B}{p}$.
		By the distributivity of the sum over the product, it is sufficient to show the property in the case of connected $A$ and $B$. 
		Recall that, by the definition of product of FDDS, the product of two connected FDDS $U,V$ has $\gcd(\cycle(U),\cycle(V))$ connected component with cycle length $\lcm(\cycle(U), \cycle(V))$.  
		We deduce that all elements in $\setDive{AB}{p}$ are generated by the product of elements in $\setDive{A}{p}$ and $\setDive{B}{p}$.
		Hence $\setDive{AB}{p} \subseteq \setDive{A}{p}\setDive{B}{p}$.
		Furthermore, by the definition of $\setDive{A}{p}$, each element of $\setDive{A}{p}$ has a cycle length dividing $p$.
		Thus, the cycle length of the product of elements of $\setDive{A}{p}$ and $\setDive{B}{p}$ is bounded by $p$.
		Since the $\lcm$ of integers dividing $p$ is also a divisor of $p$, we have $\setDive{A}{p} \setDive{B}{p} \subseteq \setDive{AB}{p}$.
	\end{proof}

	Our goal is to prove, by a simple induction, the core of the proof of Proposition~\ref{prop:condSufInj}. % , namely Lemma~\ref{lemme:condSufInjSpe}. 
	Lemma~\ref{lemme:closureProduct} is a first step to show the induction case.
	Indeed, it allows us to treat sequentially the different cycle lengths. We use the order~$\ioct$ to process sequentially connected components having the same cycle length, as shown by the following lemma.
	
	\begin{lemma}\label{lemme:structurePolyFDDS}
		Let $X = X_1 + \cdots + X_k$ be an FDDS with $k$ connected components sorted by $\ioct$ and $P = \sum_{i=1}^{m} A_i X^i$ a polynomial over FDDS without constant term and with at least one cancelable coefficient. 
		Then $B = \min_{\oct}(P(X) - P(\sum_{i=1}^{n-1} X_i))$ implies $\cycle(X_n) = \cycle(B)$ for all~$1 \le n \le k$.
	\end{lemma}
	
	\begin{proof}
	        We have $\cycle(X_n) \le \lcm(a, \cycle(X_{i}))$ for all integers $i$ with $n \le i \le k$ and $a > 0$. 
		It follows that $\cycle(X_n) \le \cycle(B)$. 
		In addition, since at least one coefficient of $P$, say $A_{c}$, is cancelable, it contains one connected component whose cycle length is one. 
		This implies that $A_{c} X_{n}^c$ contains a connected component $C$ such that $\cycle(C) = \cycle(X_n)$. 
		If $B$ is minimal, then $\cycle(X_n) \ge \cycle(B)$.
	\end{proof}
	
	We can now identify recursively and unambiguously each connected component in $X$. 
	Therefore, we can use Lemma~\ref{lemme:structurePolyFDDS}, as part of the induction step for the following lemma.
	
	\begin{lemma}\label{lemme:condSufInjSpe}
		Let $P = \sum_{i=1}^{m} A_i X^{i}$ be a polynomial without constant term.
		If at least one coefficient of $P$ is cancelable then $P$ is injective.
	\end{lemma}
	
	\begin{proof}
	        Let $X = X_1 + \cdots + X_{n_1}$, resp., $Y = Y_1 + \cdots + Y_{n_2}$ be FDDS consisting of~$n_1$ (resp., $n_2$) connected components sorted by~$\ioct$, and let~$B$ be an FDDS. Suppose~$P(X) = B = P(Y)$.
		We prove by induction on the number~$n_1$ of connected components in $X$ that $X = Y$.
		
		If $X$ has $0$ connected components, since $P$ does not have constant term, we have $P(X) = 0$.
		In addition, since $P$ has at least a cancelable coefficient, it trivially contains at least one nonzero coefficient.
		Thus, since $P(Y) = P(X) = 0$, we deduce that $Y = 0$, hence $X = Y$.
		The property is true in the base case.
		
		Let $n$ be a positive integer.
		Suppose that the properties is true for $n$, \ie $X_1 + \cdots + X_n = Y_1 + \cdots + Y_n$.
		Then, $P(X_1 + \cdots + X_n) = P( Y_1 + \cdots + Y_n)$.
		This implies that $P(X) - P(X_1 + \cdots + X_n) = B - P(X_1 + \cdots + X_n) = P(Y) -  P(Y_1 + \cdots + Y_n)$.
		Let $p$ be the length of the cycle of the smallest connected component of $B - P(X_1 + \cdots + X_n)$ according to $\oct$.
		By Lemma \ref{lemme:structurePolyFDDS}, 
		we have $\cycle(X_{n+1}) = \cycle(Y_{n+1}) = p$.
		
		Besides, thanks to Lemma \ref{lemme:closureProduct}, we deduce that $\setDive{P(X)}{p} = \sum_{i=1}^{m} \setDive{A_i}{p} \setDive{X^i}{p}$ and $\setDive{P(Y)}{p} = \sum_{i=1}^{m} \setDive{A_i}{p} \setDive{Y^i}{p}$.
		Since $\setDive{P(X)}{p} = \setDive{B}{p} = \setDive{P(Y)}{p}$, we deduce that $ \sum_{i=1}^{m} \setDive{A_i}{p} \setDive{X^i}{p} =  \sum_{i=1}^{m} \setDive{A_i}{p} \setDive{Y^i}{p}$.
		By taking unrolls, we obtain $\sum_{i=1}^{m} \unroll{\setDive{A_i}{p}} \unroll{\setDive{X^i}{p}} =  \sum_{i=1}^{m} \setDive{A_i}{p} \setDive{Y^i}{p}$.
		However, by injectivity of polynomials over unrolls (Theorem \ref{th:injPolyUnrolls}), we have $\unroll{\setDive{X}{p}} = \unroll{\setDive{Y}{p}}$.
		Thus, the smallest tree of $\unroll{\setDive{X}{p}} - \unroll{\setDive{X_1 + \cdots + X_n}{p}}$ is equal to the smallest of $\unroll{\setDive{Y}{p}} - \unroll{\setDive{Y_1 + \cdots + Y_n}{p}}$.
		However, the smallest tree of $\unroll{\setDive{X}{p}} - \unroll{\setDive{X_1 + \cdots + X_n}{p}}$ is the smallest tree of $\unroll{X_{n+1}}$ and of $\unroll{Y_{n+1}}$.
		This implies that $X_{n+1}$ and $Y_{n+1}$ are two connected components whose cycle length is $p$ and have the same minimal unroll tree, so $X_{n+1} = Y_{n+1}$.
		The induction step and the statement follow.
	\end{proof}
	
	What remains in order to conclude the proof of Proposition \ref{prop:condSufInj} is to extend Lemma \ref{lemme:condSufInjSpe} to polynomials with a constant term.
	
	\begin{proof}[Proof of Proposition \ref{prop:condSufInj}]
		Let $P = \sum_{i=0}^{m} A_i X^{i}$ be a polynomial with at least one non-constant cancelable coefficient.
		Let $P' = \sum_{i=1}^{m} A_i X^{i}$ be the polynomial obtained from $P$ by removing the constant term $A_0$.
		Hence $P = P' + A_0$.
		If there exist two FDDS $X,Y$ such that $P(X) = P(Y)$, then $P'(X) + A_0 = P'(Y) + A_0$.
		This implies that $P'(X) = P'(Y)$.
		Thus, by Lemma \ref{lemme:condSufInjSpe}, we conclude that $X = Y$. 
	\end{proof}
	
	From the proof of Lemma~\ref{lemme:structurePolyFDDS} and Theorem~\ref{theorem:equation_forest_P}, we deduce that we can solve in polynomial time all equations of the form~$P(X) = B$ with~$P$ having a cancelable non-constant coefficient. 
	Indeed, we can just select the connected components with the correct cycle lengths, and cut their unrolls to the correct depth. 
	Then, we solve the equation over the forest thus obtained, select the smallest tree of the result, and re-roll it to the correct period. 
	Finally, we remove the corresponding multiset of connected components and reiterate until either all connected components have been removed, or we find a connected component that we cannot remove (which implies that the equation has no solution). 
	
	%%% Local Variables:
	%%% mode: LaTeX
	%%% TeX-master: "main"
	%%% End:

	\section{Necessary condition for the injectivity of polynomials}\label{section:cond_nec_poly_FDDS_inj}
	
	We will now show that the sufficient condition for injectivity of Proposition \ref{prop:condSufInj} is actually necessary.
	As previously, we start by considering polynomials without constant terms.
	Our starting point for this proof is given by  \cite[Theorem 34]{article_arbre}, which characterizes the injectivity of linear monomials.
	Furthermore \cite[Lemma 33]{article_arbre} proves that this condition is necessary for linear monomials.
	We begin by extending the proof of \cite[Lemma 33]{article_arbre} to show that this is necessary for \emph{all} monomials. 
	
	For this, let $\delta_J$ be the sequence recursively defined by $\delta_\emptyset = 1$ and $\delta_{J \cup \{a\}} = \gcd(a, \lcm(J)) \delta_J$ for all $a \in \N$. 
	Now let~$\mathcal{A}$ be a set of integers strictly larger than $1$ and for $I \subseteq \mathcal{A}$ we define $\alpha_I = \delta_\mathcal{A}\prod_{a \in \mathcal{A}} a$ and $\beta_I = \alpha_I + (-1)^{|I|}\delta_I \prod_{a \in \mathcal{A} - I} a$.
	
	We will use $\alpha_I$ and $\beta_I$ to construct two different FDDS $X$ and $Y$ such that $AX = AY$ if $A$ is not cancelable.
	We will start with the case where $A$ is a cycle, denoted by $C_b$ where $b = \cycle(C_b)$.
	% We First set a technical lemma.
	
	\begin{lemma}\label{lemme:cycleNonInjTech}
		% Let $\mathcal{A}$ be a set of integers strictly larger than $1$ and
                Let $k \ge 1$ be an integer, $b \in \mathcal{A}$, $I \subseteq \mathcal{A} - \{b\}$ and $J = I \cup \{b\}$. Then
                \begin{equation}
                \label{eq:kToKMinus1}
                C_b (\beta_I C_{\lcm(I)} + \beta_J C_{\lcm(J)})^k = C_b (\alpha_I C_{\lcm(I)} + \alpha_J C_{\lcm(J)})^k.
                \end{equation}
	\end{lemma}
	
	\begin{proof}
		% We have
		% \begin{equation}
		% 	C_b (\beta_I C_{\lcm(I)} + \beta_J C_{\lcm(J)})^k = C_b (\beta_I C_{\lcm(I)} + \beta_J C_{\lcm(J)}) (\beta_I C_{\lcm(I)} + \beta_J C_{\lcm(J)})^{k-1}.
		% \end{equation}
		The proof of \cite[Lemma 33]{article_arbre} shows that $C_b (\beta_I C_{\lcm(I)} + \beta_J C_{\lcm(J)}) = C_b (\alpha_I C_{\lcm(I)} + \alpha_J C_{\lcm(J)})$.
		Thus, we can replace one instance $C_b (\beta_I C_{\lcm(I)} + \beta_J C_{\lcm(J)})$ in \eqref{eq:kToKMinus1} % by $C_b (\alpha_I C_{\lcm(I)} + \alpha_J C_{\lcm(J)})$
                and obtain:
		\[C_b (\beta_I C_{\lcm(I)} + \beta_J C_{\lcm(J)})^k =(\alpha_I C_{\lcm(I)} + \alpha_J C_{\lcm(J)}) C_b (\beta_I C_{\lcm(I)} + \beta_J C_{\lcm(J)})^{k-1}.\]
		By repeating the same substitution, the thesis follows.
	\end{proof}
	
	Now we can prove that monomials of the form $A X^k$, with $k \ge 1$ and where $A$ is a sum of cycles (\ie, a
        \emph{permutation}) with no cycle of length~$1$, are never injective.
	
	\begin{lemma}\label{lemme:noInjMonomeP1}
		% Let $\mathcal{A}$ be a set of integer strictly superior to $1$ and
                Let~$A = A_1 + \cdots + A_m$ be a sum of cycles of length greater than~$1$ and $k \ge 1$ be an integer.
		Then, there exist two different $FDDS$ $X,Y$ such that $A_j X^k = A_j Y^k$ for all~$1 \le j \le m$ and, by consequence, $A X^k = A Y^k$.
	\end{lemma}

	\begin{proof}
                Let~$\mathcal{A}$ be the set of cycle lengths of~$A$.      
                Let $X = \sum_{I \subseteq \mathcal{A}} \alpha_I C_{\lcm(I)}$ and $Y = \sum_{I \subseteq \mathcal{A}} \beta_I C_{\lcm(I)}$. Remark that~$\alpha_I \ne \beta_I$ for all~$I \subseteq \mathcal{A}$, since~$\beta_I$ is $\alpha_I$ plus a nonzero term; in particular, $\alpha_\emptyset \ne \beta_\emptyset$ and, since these two integers are the number of fixed points of~$X$ and~$Y$ respectively, we have~$X \ne Y$.

                Let us consider a generic term $A_j$, which is a cycle $C_b$ for some $b>1$ for all~$j$. Let~$I_1, \ldots, I_n$ be an enumeration of the subsets of~$\mathcal{A}-\{b\}$ and let~$J_i = I_i \cup \{b\}$ for all~$1 \le i \le n$. Then~$Y = \sum_{i=1}^n (\beta_{I_i} C_{\lcm(I_i)} + \beta_{J_i} C_{\lcm(J_i)})$
		
		% We have that $C_b Y^k = C_b (\sum_{I \subseteq \mathcal{A} - \{b\}} \beta_I C_{\lcm(I)} + \beta_J C_{\lcm(J)})^k$.
		In order to compute~$A_j Y^k = C_b Y^k$ we can distribute the power over the sum, then the product over the sum, and we obtain:
		\begin{equation}\label{eq:distProdSum}
			C_b Y^k = \sum_{k_1+\cdots+k_n=k}\binom{k}{k_1,\ldots,k_n} C_b (\beta_{I_1} C_{\lcm(I_1)} + \beta_{J_1} C_{\lcm(J_1)})^{k_1}\prod_{i=2}^{n} (\beta_{I_i} C_{\lcm(I_i)} + \beta_{J_i} C_{\lcm(J_i)})^{k_i}.
		\end{equation}
		By Lemma \ref{lemme:cycleNonInjTech}, we have $C_b (\beta_{I_1} C_{\lcm(I_1)} + \beta_{J_1} C_{\lcm(J_1)})^{k_1} = C_b (\alpha_{I_1} C_{\lcm(I_1)} + \alpha_{J_1} C_{\lcm(J_1)})^{k_1}$.
		Thus, by replacing this in Equation \eqref{eq:distProdSum}, we have
		\[C_b Y^k = \sum_{k_1+\cdots+k_n=k}\binom{k}{k_1,\ldots,k_n} C_b (\alpha_{I_1} C_{\lcm(I_1)} + \alpha_{J_1} C_{\lcm(J_1)})^{k_1} \prod_{i=2}^{n} (\beta_{I_i} C_{\lcm(I_i)} + \beta_{J_i} C_{\lcm(J_i)})^{k_i}.\]
		If we repeat this substitution for all $2 \le i \le n$ and factor $C_b$ from each term of the sum we obtain
		\[C_b Y^k = C_b \sum_{k_1+\cdots+k_n=k}\binom{k}{k_1,\ldots,k_n}  \prod_{i=1}^{n} (\alpha_{I_i} C_{\lcm(I_i)} + \alpha_{J_i} C_{\lcm(J_i)})^{k_i}.\] 
		Since $\sum_{k_1+\cdots+k_n=k}\binom{k}{k_1,\ldots,k_n}  \prod_{i=1}^{n} (\alpha_{I_i} C_{\lcm(I_i)} + \alpha_{J_i} C_{\lcm(J_i)})^{k_i}$ is just $X^k$, we conclude that $C_b Y^k = C_b X^k$, \ie, $A_j Y^k = A_j X^k$.

                This holds separately for each connected component~$A_j$ of~$A$; by adding all terms and factoring $X^k$ and $Y^k$ we obtain $AX^k = AY^k$.
	\end{proof}
	
	We remark that the $X$ and $Y$ which we have built are just permutations. 
	An interesting property of permutations is that they are closed under sums and products (and thus nonnegative integer powers).
	Thus $X^k$ and $Y^k$ are also permutations.
	Another useful property of permutations is, intuitively, that if we multiply a sum of connected components $A = A_1 + \cdots + A_n$ with a sum of cycles $X = X_1 + \cdots + X_m$, the trees rooted in each connected component of $A_i X_j$ are just a periodic repetition of the trees rooted in the cycle of $A_i$ for all $i,j$ \cite[Corollary 14]{article_arbre}.
	Thanks to these properties, we can characterize the injectivity of monomials.
	
	\begin{proposition}\label{prop:charaInjMonome}
		A monomial over FDDS is injective if and only if its coefficient is cancelable.
	\end{proposition}
	
	\begin{proof}
		Assume that the monomial $AX^k$, where $A = A_1 + \cdots + A_n$ as a sum of connected components and~$k>0$, is not cancelable. Hence, by \cite[Theorem 34]{article_arbre}, $A$ does not have a cycle of length 1 (a fixed point).
		Let $B = B_1 + \cdots + B_n$ the permutation such that~$\cycle(B_i) = \cycle(A_i)$ for all~$i$.
		By Lemma \ref{lemme:noInjMonomeP1}, there exist two distinct permutations $X,Y$ such that $B_i X^k = B_i Y^k$ for all $i$. 
		Thus, by~\cite[Corollary 14]{article_arbre}, we have $A_i X^k = A_i Y^k$ and, by summing, $A X^k = A Y^k$.

		The inverse implication was already proved in Proposition \ref{prop:condSufInj}.
        \end{proof}
	
	In the rest of this section, we will prove that a polynomial is injective if and only if a certain monomial is injective.
	Remark that the sum of any FDDS with a cancelable FDDS is also cancelable.
	Hence, the sum of the coefficients of a polynomial is cancelable if and only if at least one of them is cancelable.
	
	From the characterization of injective monomials (Proposition~\ref{prop:charaInjMonome}) we can deduce that there exists an positive integer $k$ such that $A X^k$ is injective if and only if $A X^n$ is injective for \emph{all} positive integer $n$.
	
	\begin{proposition}\label{prop:charaInj}
		Let $P = \sum_{i=1}^{m} A_i X^{i}$ be a polynomial without constant term and consider the monomial~$M = (\sum_{i=1}^{m} A_i)X$. Then, $P$ is injective if and only if $M$ is injective.
	\end{proposition}
	
	\begin{proof}
		Assume that $M$ is injective.
		Thus, by the Proposition \ref{prop:charaInjMonome}, its coefficient is cancelable.
		Hence, at least one coefficient of $P$ is cancelable.
		By Proposition \ref{prop:condSufInj}, we conclude that $P$ is injective. 
		
		For the other direction, assume that $M$ is not injective.
		Let $X$ and $Y$ the FDDS constructed in the proof of Proposition \ref{prop:charaInjMonome}.
		As previously, $X \neq Y$ but $A_i X^i = A_i Y^i$ for all $i \in \{1,\ldots, m\}$. 
		So, by sum, we have $P(X) = P(Y)$.
	\end{proof}
	
	Thanks to this characterization, we conclude that the condition given in Proposition \ref{prop:condSufInj} is also necessary for polynomials without a constant term.
	All it remains to prove is that this is also necessary in the presence of a constant term.
	
	\begin{theorem}\label{th:charaInjPoly}
		Let $P = \sum_{i=0}^{m} A_i X^{i}$ a polynomial. Then $P$ is injective if and only if at least one of $A_i$ is cancelable for some~$i \ge 1$.
	\end{theorem}
	
	\begin{proof}
                Assume that no non-constant coefficient of $P$ is cancelable.
		Let $P' = \sum_{i=1}^{m} A_i X^{i}$ the polynomial $P$ without the term $A_0$.
		By the Proposition \ref{prop:charaInj}, there exist two different FDDS $X,Y$ such that $P'(X) = P'(Y)$.
		Then $P'(X) + A_0 = P'(Y) + A_0$, which implies $P(X) = P(Y)$. 
		Thus, $P$ is not injective.

		The inverse implication is due to Proposition \ref{prop:condSufInj}.
	\end{proof}
	%%% Local Variables:
	%%% mode: LaTeX
	%%% TeX-master: "main"
	%%% End:

	\section{Conclusion}

In this article we proved that all univariate polynomials over unrolls are injective and we characterized the injectivity of univariate polynomials over FDDS. 
Furthermore, thanks to the proof of Lemma \ref{lemme:structurePolyFDDS}, we can use a technique similar to \cite{kroot} to compute a solution (if it exists) to injective polynomial equations over unrolls and FDDS. 
This result suggests interesting directions for future work,
namely, characterizing the injectivity of multivariate polynomials, and finding efficient algorithms for other classes of equations involving non-injective polynomials (or the reason why they do not exist, if this is the case).

Furthermore, in the literature, there exists a procedure for solving and enumerating the solutions of an equation of the form $P(X_1, \ldots, X_n) = B$ where polynomial $P$ is a sum of univariate monomials \cite{dennunzio2023dds_journal}.
Thus, we can hopefully find an efficient algorithm for the special case where $P$ is a sum of injective univariate monomials. Since polynomial equations with an unbounded number of variables are $\NP$-hard~\cite{NPhard}, one could try to improve these results for equations with a bounded number of variables, as well as finding exactly how many variables are necessary for the problem to become $\NP$-hard.

%%% Local Variables:
%%% mode: LaTeX
%%% TeX-master: "main"
%%% End:

	\begin{credits}
        \subsubsection{\ackname}
        This work has been partially supported by the HORIZON-MSCA-2022-SE-01 project 101131549 ``Application-driven Challenges for Automata Networks and Complex Systems (ACANCOS)'' and by the project ANR-24-CE48-7504 ``ALARICE''.
	\end{credits}

	\bibliography{biblio}

\begin{thebibliography}{10}

\bibitem{gershenson2004random_bn}
Carlos Gershenson.
\newblock Introduction to random boolean networks.
\newblock {\em arXiv e-prints}, 2004.

\bibitem{automata_book}
Eric Goles and Servet Mart\`inez.
\newblock {\em Neural and Automata Networks: Dynamical Behavior and
  Applications}.
\newblock Kluwer Academic Publishers, 1990.

\bibitem{thomas1990biological_feedback}
Ren\'e Thomas and Richard D'Ari.
\newblock {\em Biological Feedback}.
\newblock CRC Press, 1990.

\bibitem{thomas1973genetic_control_circuits}
Ren\'e Thomas.
\newblock Boolean formalization of genetic control circuits.
\newblock {\em Journal of Theoretical Biology}, 42(3):563--585, 1973.

\bibitem{bernot2013computational_biology}
Gilles Bernot, Jean-Paul Comet, Adrien Richard, Madalena Chaves, Jean-Luc
  Gouz\'e, and Fr\'ed\'eric Dayan.
\newblock Modeling in computational biology and biomedicine.
\newblock In {\em Modeling and Analysis of Gene Regulatory Networks: A
  Multidisciplinary Endeavor}, pages 47--80. Springer, 2012.

\bibitem{reaction_systems}
Andrzej Ehrenfeucht and Grzegorz Rozenberg.
\newblock Reaction systems.
\newblock {\em Fundamenta Informaticae}, 75:263--280, 2007.

\bibitem{gadouleau2011graph_entropy}
Maximilien Gadouleau and Soren Riis.
\newblock Graph-theoretical constructions for graph entropy and network coding
  based communications.
\newblock {\em IEEE Transactions on Information Theory}, 57(10):6703–6717,
  2011.

\bibitem{article_fondateur}
Alberto Dennunzio, Valentina Dorigatti, Enrico Formenti, Luca Manzoni, and
  Antonio~E. Porreca.
\newblock Polynomial equations over finite, discrete-time dynamical systems.
\newblock In {\em Cellular Automata, 13th International Conference on Cellular
  Automata for Research and Industry, ACRI 2018}, volume 11115 of {\em Lecture
  Notes in Computer Science}, pages 298--306. Springer, 2018.

\bibitem{testIsoLinear}
John~E Hopcroft and Jin-Kue Wong.
\newblock Linear time algorithm for isomorphism of planar graphs (preliminary
  report).
\newblock In {\em Proceedings of the sixth annual ACM symposium on Theory of
  computing}, pages 172--184, 1974.

\bibitem{article_arbre}
\'Emile Naquin and Maximilien Gadouleau.
\newblock Factorisation in the semiring of finite dynamical systems.
\newblock {\em Theoretical Computer Science}, 998:114509, 2024.

\bibitem{kroot}
François Doré, Kevin Perrot, Antonio~E. Porreca, Sara Riva, and Marius
  Rolland.
\newblock Roots in the semiring of finite deterministic dynamical systems.
\newblock {\em arXiv e-prints}, 2025.
\newblock \url{https://arxiv.org/abs/2405.09236}.

\bibitem{livreGraphe}
Richard Hammack, Wilfried Imrich, and Sandi Klav\v{z}ar.
\newblock {\em Handbook of Product Graphs}.
\newblock Discrete Mathematics and Its Applications. CRC Press, second edition,
  2011.

\bibitem{dennunzio2023dds_journal}
Alberto Dennunzio, Enrico Formenti, Luciano Margara, and Sara Riva.
\newblock An algorithmic pipeline for solving equations over discrete dynamical
  systems modelling hypothesis on real phenomena.
\newblock {\em Journal of Computational Science}, 66:101932, 2023.

\bibitem{NPhard}
Antonio~E. Porreca.
\newblock Composing behaviours in the semiring of dynamical systems.
\newblock In {\em International Workshop on Boolean Networks (IWBN 2020)},
  2020.
\newblock Invited talk, \url{https://doi.org/10.5281/zenodo.3934396}.

\end{thebibliography}
	\bibliographystyle{unsrt}
	
\end{document}